\newif\ifabstract
\abstracttrue
\newif\iffull
\ifabstract \fullfalse \else \fulltrue \fi

\documentclass[11pt]{article}
\usepackage{amsmath,amsfonts,amssymb,amsthm}
\usepackage{fullpage}
\usepackage{algorithm}
\usepackage{algorithmic}
\usepackage{epstopdf}
\usepackage{graphicx}
\usepackage{hyperref}
\usepackage{cite}
\usepackage{tikz}
\usepackage{enumitem}

\newtheorem{theorem}{Theorem}
\newtheorem{heuristic algorithm}{Heuristic Algorithm}
\newtheorem{lemma}{Lemma}
\newtheorem{definition}{Definition}
\newtheorem{corollary}{Corollary}

\newtheorem{proposition}{Proposition}

\newtheorem{example}{Example}

\newtheorem{conjecture}{Conjecture}

\newtheorem{approximation algorithm}{Aprroximation Algorithm}

\newcommand{\R}{\mathbb{R}}

\newcommand{\vol}{\mathrm{vol}}

\newcommand{\Con}{\Psi_n}
\newcommand{\UC}{\mathcal{BT}_n}
\newcommand{\NC}{\mathcal{NC}_n}
\newcommand{\RF}{\mathcal{RF}_n}
\newcommand{\Conb}{\Psi}
\newcommand{\UCb}{\mathcal{BT}}

\newcommand{\SK}{\mathsf{SK}}

\newcommand{\indicator}{\mathbf{1}}

\newcommand{\Network}{\mathcal{N}}
\newcommand{\Conv}{\mathsf{Conv}}

\newcommand{\calA}{\mathcal{A}}
\newcommand{\calB}{\mathcal{B}}
\newcommand{\calS}{\mathcal{S}}
\newcommand{\calT}{\mathcal{T}}

\renewcommand{\span}{\mathrm{Span}}
\newcommand{\LP}{\mathsf{LP}}

\newcommand{\eat}[1]{}

\begin{document}
\bibliographystyle{plain}

\title{On the Inequalities of Projected Volumes  and the Constructible Region}

\author{Zihan Tan\thanks{Department of Computer Science, University of Chicago. Email: {\tt zihantan@uchicago.edu}.}\and Liwei Zeng\thanks{Department of Industrial Engineering and Management Sciences, Northwestern University. Email: {\tt liweizeng2015@u.northwestern.edu}.}}



\begin{titlepage}
\thispagestyle{empty}
\maketitle
\begin{abstract}
We study the following geometry problem: given a $2^n-1$ dimensional vector $\pi=\{\pi_S\}_{S\subseteq [n], S\ne \emptyset}$, is there
an object $T\subseteq\R^n$ such that $\log(\vol(T_S))= \pi_S$, for all $S\subseteq [n]$, where
$T_S$ is the projection of $T$ onto the subspace spanned by the axes in $S$ and $\vol(T_S)$ is its $|S|$-dimensional volume?
If $\pi$ does correspond to an object in $\R^n$, we say that $\pi$ is {\em constructible}.
We use $\Con$ to denote the constructible region, i.e., the set of all constructible vectors in $\R^{2^n-1}$.
In 1995, Bollob\'{a}s and Thomason showed that $\Con$ is contained in a polyhedral cone and defined a class of so called
uniform cover inequalities.
We propose a new set of inequalities, called nonuniform-cover inequalities, which generalizes the uniform cover inequalities.
We show that any linear inequality that all points in $\Con$ satisfy must be a nonuniform-cover inequality.
Based on this result and an example by Bollob\'{a}s and Thomason,
we show that the constructible region $\Con$ is non-convex for $n\geq 4$, and thus cannot be fully characterized by linear inequalities.
We further show that some subclasses of the nonuniform-cover inequalities are not satisfied by all constructible vectors via 
various combinatorial constructions, which refutes a previous conjecture about $\Con$.
Finally, we conclude with an interesting conjecture regarding the convex hull of $\Con$.
\end{abstract}

\thispagestyle{empty}

\end{titlepage}

\section{Introduction}
\label{sec:intro}

Let $T$ be an object in $\mathbb{R}^n$ and let $\{e_1,\cdots,e_n\}$ be the standard basis of $\mathbb{R}^n$.
By an object, we mean a compact subset of $\mathbb{R}^n$.
Let $\span(S)$ be the subspace spanned by $\{e_i\mid i\in S\}$.
Given an index set $S\subseteq[n]=\{1,2,\cdots,n\}$ with $|S|=d$, we denote by $T_S$ the orthogonal projection of $T$ onto $\span(S)$, and by $|T_S|$ its $d$-dimensional volume.
We use $|T|$ to denote the $n$-dimensional volume of $T$.
Given an $n$-dimensional object $T$, define $\pi(T)$ to be the {\em log-projection vector} of $T$,
which is a $2^n-1$ dimensional vector with entries indexed by all nonempty subsets of $[n]$ and $\pi(T)_{S}=\log |T_S|$ for all $S\subseteq [n], S\neq \emptyset$\footnote{In this paper all logarithms are to the base of $2$, and we use the convention that $\log 0 = -\infty$.}.
Whenever we refer to a $2^n-1$ dimensional vector $\pi$, we assume that the entries are indexed by the nonempty subsets of $[n]$
(i.e., $\pi_S$ is the entry indexed by $S\subseteq [n]$).
We say that a $2^n-1$ dimensional vector $\pi$ is {\em constructible} if $\pi$ is the log-projection vector of an object $T$ in $\R^n$.
We define the constructible region $\Con$, the central subject studied in this paper, to be the set of all constructible vectors:
$$
\Con=\bigg\{\pi\in \mathbb{R}^{2^n-1}\text{ }\bigg| \text{  } \pi\text{ is constructible}\bigg\}.
$$
With above definitions, it is natural to ask the following questions:
\begin{enumerate}
\item Given a $2^n-1$ dimensional vector $\pi$, is there an algorithm to decide whether $\pi$ is in $\Con$?
\item What is the geometric structure of $\Con$? What properties does $\Con$ have?
\end{enumerate}

In 1995, Bollob\'{a}s and Thomason~\cite{bollobas1995projections} proposed a class of inequalities
relating the projected volumes.  Their result reads as follows.
Let $\calA$ be a family of subsets of $S$.
We say $\calA$ is
a $k$-cover of $S$, if each element of $S$ appears exactly $k$ times in the multiset induced by $\calA$.
For example, $\{\{1,2\},\{2,3\},\{1,3\}\}$ is a 2-uniform cover of $\{1,2,3\}$.

\begin{theorem}(Bollob\'{a}s-Thomason (BT) uniform-cover inequalities)
\label{thm:BT}
Let $T$ be an object in $\R^n$ and let $\calA$ be a $k$-cover of $[n]$, we have
$
|T|^k\le \prod_{A\in \calA} |T_A|.
$
\end{theorem}
With the above notations, we define the polyhedron cone
$$
\UC = \bigg\{ \pi \in \R^{2^n-1}\text{ }\bigg|\text{ } k\pi_{S}\leq \sum_{A\in \calA} \pi_A, \text{ for all }(k,\calA,S)\text{ such that }S\subseteq [n]\text{ and }\calA\text{ } k\text{-covers } S \bigg\}.
$$
BT inequalities essentially assert that
every constructible vector is in $\UC$, or equivalently $\Con\subseteq \UC$.
In the very same paper~\cite{bollobas1995projections}, they also found a {\em non-constructible} vector in $\UCb_4$, which
implies that $\Con\varsubsetneq \UC$ for $n\ge 4$.
However, their results do not rule out the possibility that $\Con$ is convex, or even can be characterized by
a finite set of linear inequalities.

\subsection{Our Results}
Besides the results mentioned above, very little is known about $\Con$ and the main
goal of this paper is to deepen our understanding about its structure.
We first propose a new class of inequalities, called nonuniform-cover inequalities, which generalizes
the BT uniform-cover inequalities.
The following notations are used (throughout the paper) to define nonuniform-cover inequalities.

Let $\calA=\{A_i\}_{i=1}^{k}$, $\calB=\{B_j\}_{j=1}^{m}$ be two families of subsets\footnote{
A subset of $[n]$ may appear multiple times in $\calA$ or $\calB$.
} of $[n]$,
where $A_i$ and $B_j$ are subsets of $[n]$.
We say $\calA$ {\em covers} $\calB$ if:

\begin{enumerate}

 \item[P1.] The disjoint union of $\{A_i\}_{i=1}^{k}$ is the same as the disjoint union of $\{B_j\}_{j=1}^{m}$.
 In other words, for every element $e\in[n]$,
 $|\{i\mid e\in A_i\}| = |\{j\mid e\in B_j\}|$.

 \item[P2.] Let $\Sigma=\{(A_i,t)\mid t\in A_i\}$ and $\Lambda=\{(B_j,s)\mid s\in B_j\}$,
 there exists an one-to-one mapping $f$ between $\Sigma$ and $\Lambda$ such that:
 for any $(A_i,t)\in \Sigma$ with $(B_j,s)=f(A_i,t)$, $t=s$ and $A_i\subseteq B_j$.
\end{enumerate}

\begin{definition}(Nonuniform-Cover (NC) inequalities)
Let $x$ be a $2^n-1$ dimensional vector indexed by nonempty subsets of $[n]$ and assume that $\calA$ covers $\calB$.
A nonuniform-cover inequality is defined as:
\[\prod_{A_i\in \calA}x_{A_i}\ge \prod_{B_j\in\calB} x_{B_j}.\]
\end{definition}

\begin{example}
\label{ex:1}
Let $\calA=\{\{1,2\}, \{2,3\}, \{3,4\} \}$ and
$\calB=\{\{1,2,3\}, \{2,3,4\}\}$.
We can see $\calA$ covers $\calB$.
The corresponding NC inequality is
$x_{\{1,2\}}\cdot x_{\{2,3\}}\cdot x_{\{3,4\}} \geq x_{\{1,2,3\}}\cdot x_{\{2,3,4\}}$.
Here is another example:
$
x_{\{1\}}\cdot x_{\{1,2\}}\cdot x_{\{2,3\}}\cdot x_{\{3,4\}}\cdot x_{\{2,4\}} \geq x_{\{1,2,3\}}\cdot x_{\{2,3,4\}}\cdot x_{\{1,2,4\}}.
$
\end{example}

\noindent
When the context is clear, we refer to a linear inequality
of the form $\sum_{A_i\in\calA}\pi_{A_i}\geq \sum_{B_j\in\calB} \pi_{B_j}$
as a NC inequality as well.
We say that the NC inequality $\sum_{A_i\in\calA}\pi_{A_i}\geq \sum_{B_j\in\calB} \pi_{B_j}$ is the linear form of the NC inequality $\prod_{A_i\in \calA}x_{A_i}\ge \prod_{B_j\in\calB} x_{B_j}$.
And we say that an object $T\subseteq \mathbb{R}^n$ satisfies the NC inequality $\prod_{A_i\in \calA}x_{A_i}\ge \prod_{B_j\in\calB} x_{B_j}$ if $\prod_{A_i\in \calA}|T_{A_i}|\ge \prod_{B_j\in\calB} |T_{B_j}|$, or equivalently $\sum_{A_i\in\calA}\pi(T)_{A_i}\geq \sum_{B_j\in\calB} \pi(T)_{B_j}$ (i.e., its log-projection vector satisfies the linear form of the NC inequality). 
It is not hard to see that that every BT inequality is an NC inequality.
But the converse may not be true. For example,
$
x_{\{1,2\}}\cdot x_{\{2,3\}}\cdot x_{\{3,4\}}\ge x_{\{1,2,3\}}\cdot x_{\{2,3,4\}}.
$
(We alert the reader that we do not claim such inequalities are correct for all constructible vectors, and in fact some NC inequalities are not satisfied by all constructible vectors. We will discuss it in detail in Section~\ref{sec:counterex}.)

Similar to $\UC$, we define $\NC$ to be the set of all $2^n-1$ dimensional vectors that satisfy all NC inequalities:
Formally, it is the following polyhedral cone:
$$
\NC = \bigg\{ \pi \in \R^{2^n-1} \text{ }\bigg|\text{ } \sum_{B_j\in\calB}\pi_{B_j}\leq \sum_{A_i\in \calA} \pi_{A_i},
\text{ for all }\calA, \calB\subseteq 2^{[n]} \text{ such that } \calA \text{ covers } \calB \bigg\}.
$$
Our first result states that all correct linear inequalities must be in this class.
\begin{theorem}
\label{thm:nonuniform}
If all points in $\Con$ satisfy a linear inequality $\sum_{S\subseteq [n]} \alpha_S \pi_{S}\leq 0$,
the linear inequality must be a NC inequality, or a positive combination of NC inequalities.
\end{theorem}

In order to prove Theorem \ref{thm:nonuniform}, we introduce a class of objects called {\em rectangular flowers}.
Let $\RF$ be the set of log-projection vectors generated by all rectangular flowers
(see the definition in Section~\ref{sec:proof1}).
For any linear inequality that is not an NC inequality or a positive combination of NC inequalities,
we construct a rectangular flower that violates the inequality.
It is not hard to show that the log-projection vector of a rectangular flower in $\R^n$ satisfies
all nonuniform cover inequalities (i.e., it is in $\NC$).
Moreover, we show that for every point $\pi\in \NC$,
there exists a rectangular flow in $\R^n$ whose log-projection vector is equal to $\pi$.
Therefore, we have the following theorem.
\begin{theorem}
\label{thm:recflower}
For all $n\geq 1$, $\NC=\RF\subseteq \Con$.
\end{theorem}

Given Theorem \ref{thm:recflower}, it is natural to ask whether $\NC=\Con$.
If the answer was yes, we would have a compact description for $\Con$ and
deciding whether a point is in $\Con$ can be done using linear programming (see Section~\ref{sec:proof1} for details).
However, Theorem \ref{thm:nonconvex} shows the answer is no by proving that $\Con$ is non-convex for $n\geq 4$ (while $\NC$ is convex).
We note that $\Con = \UC$ is convex for $n\leq 3$. Proof for this claim appears in Appendix~\ref{app:threed}.

\begin{theorem} (Non-convexity of $\Con$)
\label{thm:nonconvex}
For $n\geq 4$, $\Con$ is not convex.
\end{theorem}


Theorem~\ref{thm:nonconvex} implies that there exist certain constructible vectors in $\R^{2^n-1}$ that violate
some NC inequalities. In other words, $\NC\varsubsetneq \Con$.
Thus, it would be interesting to know which NC inequalities are true and which are false
(we already know BT inequalities are true).
In Section~\ref{sec:counterex},
we provide several methods for constructing counterexamples for different subclasses of NC inequalities.
However, we are not able to disprove all NC inequalities that are not BT inequalities,
nor prove the correctness of any of them.
This leads us to the following conjecture that claims only BT inequalities are correct.

\begin{conjecture}
\label{conj:1}
If all points in $\Con$ satisfy a certain linear inequality $\sum_j \beta_j\pi_{B_j}\leq \sum_i\alpha_i\pi_{A_i}$,
the linear inequality must be a BT inequality or a positive combination of BT inequalities.
Equivalently, $\UC=\Conv(\Con)$, the convex hull of $\Con$.
\end{conjecture}






At the end of this part, we summarize our results in the following chain for $n\geq 4$:
\[
\RF = \NC \subsetneq \Con \subsetneq \Conv(\Con) \subseteq \UC,
\]
and we conjecture that $\Conv(\Con)= \UC$.

\subsection{A Motivating Problem from Databases}

Our problem is closely related to the data generation problem \cite{arasu2011data} studied in the area of databases,
which is also our initial motivation for studying the problem.
Generating synthetic relations under various constraints is a key problem for testing data management systems.
A relation $R(A_1,\ldots, A_n)$ is a table, where each row is one record about some entity, and each column $A_i$ is an attribute.
One of the most important operations in relational databases is the projection operation to a subset of attributes.
One can think of the projection to subset $S$ of attributes, denoted as $\Pi_S(R)$, as
the table $R$ first restricted to columns in $S$, and then with duplication removed.
To see the connection between the database problem and geometry,
think of a relation $R(A_1,\ldots, A_n)$ with $n$ attributes as a $n$-dimensional object $T$
in $\R^n$:
A tuple (i.e., a row) $(t_1,t_2,\ldots, t_m)$ can be thought as a unit cube
$[t_1-1,t_1]\times\ldots\times[t_m-1,t_m]$.
$T_S$, the projection of $T$ onto $\span(S)$, exactly corresponds to the projected relation $\Pi_S(R)$.

\begin{example}
The following table shows the information of course registration. $5$ items in the table correspond to unit squares in the coordinate system. In this way, a table is represented by an object in Euclidean space.

\centering
\includegraphics[scale=0.43]{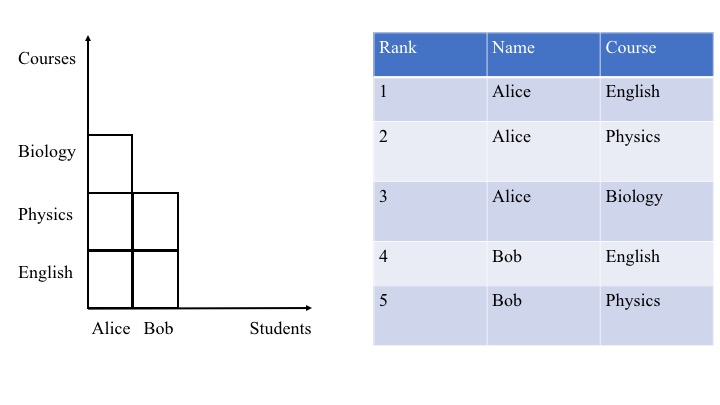}

\end{example}

In the data generating problem with projection constraints,
we are given the cardinalities $|\Pi_S(R)|$ for a set of subsets $S\subseteq [n]$.
The goal is to construct a relation $R$ that is consistent with the given cardinalities, which is a discrete version of our geometry problem.
Moreover, if the given cardinalities (after taking logarithm) is not in $\Con$, or violate any projection inequality,
there is no solution to the data generation problem.
Therefore, a good understanding of the geometric structure of $\Con$ is central for solving the data generation problem.

\subsection{Other Related Work}

Loomis and Whitney proved a class of projection inequalities in \cite{loomis1949inequality}, 
allowing one to upper bound the volume of a $d-$dimensional object by its $(d-1)$-dimensional projection volumes.
Their inequalities are a subclass of BT inequalities.
BT inequalities and their generalizations also play an essential role in 
the worst-case optimal join problem in databases
(we can get an upper bound of the size of the relation $R$ knowing the cardinalities of its projections).
See \cite{ngo2012worst} for some recent results on this problem.

There is a large body of literature on the constructible region $\Gamma_n$ for the joint entropy function over $n$ random variables $X_1,\ldots, X_n$.
More specifically, for each joint distribution over $X_1,\ldots, X_n$, there is a point in $\Gamma_n$, which is a $2^n-1$ dimensional vector, 
with the entry indexed by $S\subseteq[n]$ being $H(\{X_i\}_{i\in S})$.  
Characterizing $\Gamma_n$ for $n\ge 3$ is a major problem in information theory and has been studied extensively.
Many entropy inequalities are known, including Shannon-type inequalities and non-Shannon-type inequalities \cite{zhang1997non,makarychev2002new,zhang2002new,matus2007infinitely}.
For a comprehensive treatment of this topic, we refer interested readers to the book~\cite{yeung2008information}. 
There are concrete connections between entropy inequalities and projection inequalities \cite{chung1986some,bollobas1995projections, friedgut2004hypergraphs,balister2012projections}.
In particular, BT inequalities can be easily derived from the well-known Shearer's entropy inequalities \cite{chung1986some}   
(many even regard them as the same).

\section{Proof of Theorem~\ref{thm:nonuniform} and Theorem~\ref{thm:recflower}}
\label{sec:proof1}

In this section, we prove Theorem \ref{thm:nonuniform} and Theorem \ref{thm:recflower}.
We introduce a class of geometric objects that are crucial to our proofs.
We say a $n$-dimensional object $T\subseteq \R^n_+=\{x\mid x_i\ge 0, \forall i\in [n]\}$ is {\em cornered} if
$x\in T$ implies $y\in T$ for all $\vec{0}\le y\leq x$ (i.e., $0\le y_i\leq x_i$ for all $i\in [n]$).
An object $T\subseteq \mathbb{R}^n_+$ is said to be an {\em open rectangle} if $T= (0,a_1]\times (0,a_2]\times\ldots\times (0,a_n]$,
or a {\em closed rectangle} if $T= [0,a_1]\times [0,a_2]\times\ldots\times [0,a_n]$.

\begin{definition}
We say $T\subseteq \mathbb{R}^n_+$ is a {\em rectangular flower} if (i) $T$ is {\em cornered}, and (ii) $T\cap (0,\infty)^S$ is an open rectangle in $(0,\infty)^S$ for any $S\subseteq [n]$.
\end{definition}

\begin{figure}[h]
	\centering
	\includegraphics[width=0.7\linewidth]{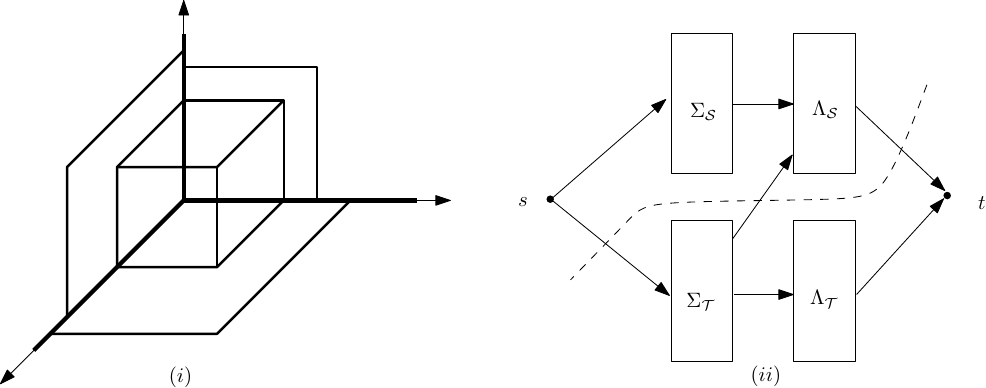}
	\caption{(i) A 3-dimensional rectangular flower.
		(ii) The network flow $\Network(\calA,\calB)$. The dashed line represents the minimum $s$-$t$ cut.}
	\label{fig:flower}
\end{figure}

See Figure~\ref{fig:flower} for an example.
It is easy to see that a rectangular flower $T\subseteq \R^n_+$ is the union of $2^n-1$ closed rectangles
$\bigcup_{S\subseteq [n], S\ne \emptyset} T_S$, each $T_S$ being a closed rectangle in $\span(S)$.
Moreover, if $S\subseteq S'$, for any $i\in S$, the edge length of $T_S$ along axis $i$ is at least that of $T_{S'}$ (since $T$ is cornered).

We then introduce a new class of inequalities, called fractional nonuniform-cover inequalities, which can be seen as
the fractional generalization of NC inequalities.
Let $\calA=\{(A_i,\alpha_i)\}_{i=1}^{k}$, $\calB=\{(B_j,\beta_j)\}_{j=1}^{m}$ be two families of weighted subsets of $[n]$,
where $A_i$ and $B_j$ are subsets of $[n]$ and $\alpha_i>0$ ($\beta_j>0$ resp.) is the positive weight associated with $A_i$ ($B_j$ resp.).
We construct a network flow instance $\Network(\calA,\calB)$ as follows:
let $\Sigma=\{(A_i,x)\mid x\in A_i, A_i\in \calA\}$ and $\Lambda=\{(B_j,y)\mid y\in B_j, B_j\in \calB\}$
be two sets of nodes.
Let node $s$ be the source and node $t$ be the sink.
There is an arc from $s$ to each node $(A_i,x)\in \Sigma$ with capacity $\alpha_i$.
There is an arc from each node $(B_j,y)\in \Lambda$ to $t$ with capacity $\beta_j$.
For each pair of $(A_i,x)$ and $(B_j,y)$, there is an arc with capacity $+\infty$ from $(A_i,x)$ to $(B_j,y)$
if $A_i\subseteq B_j$ and $x=y$.
We say $\calA$ {\em saturates} $\calB$ if the following properties hold:
\begin{enumerate}
	\item[C1.]
	For any $x\in [n]$, $\sum_{i=1}^k \alpha_i \cdot \mathbf{1}[x\in A_i]=\sum_{j=1}^m \beta_j \cdot\indicator[x\in B_j]$.
	
	\item[C2.]
	The maximum $s$-$t$ flow (or equivalently, the minimum $s$-$t$ cut) of $\Network(\calA,\calB)$ is $\sum_j \beta_j |B_j|$.
\end{enumerate}

\begin{definition}(Fractional-Nonuniform-Cover (FNC) inequalities)
	Suppose $x$ is a $2^n-1$ dimensional vector indexed by nonempty subsets of $[n]$ and $\calA$ saturates $\calB$.
	A fractional-nonuniform-cover inequality is of the following form:
	\[\prod_{(A_i,\alpha_i)\in \calA}x_{A_i}^{\alpha_i}\ge \prod_{(B_j,\beta_j)\in\calB}x_{B_j}^{\beta_j}.\]
\end{definition}

When the context is clear, we also refer to linear inequalities of the form
$\sum_{(A_i,\alpha_i)\in \calA} \alpha_i \pi_{A_i} \ge \sum_{(B_j,\beta_j)\in\calB} \beta_j\pi_{B_j}$
as FNC inequalities.
We say that $\sum_{(A_i,\alpha_i)\in \calA} \alpha_i \pi_{A_i} \ge \sum_{(B_j,\beta_j)\in\calB} \beta_j\pi_{B_j}$ is the linear form of the FNC inequality $\prod_{(A_i,\alpha_i)\in \calA}x_{A_i}^{\alpha_i}\ge \prod_{(B_j,\beta_j)\in\calB}x_{B_j}^{\beta_j}$.
And an object $T\subseteq \mathbb{R}^n$ satisfies the FNC inequality $\prod_{(A_i,\alpha_i)\in \calA}x_{A_i}^{\alpha_i}\ge \prod_{(B_j,\beta_j)\in\calB}x_{B_j}^{\beta_j}$ if $\prod_{(A_i,\alpha_i)\in \calA}|T_{A_i}|^{\alpha_i}\ge \prod_{(B_j,\beta_j)\in\calB}|T_{B_j}|^{\beta_j}$, or equivalently $\sum_{(A_i,\alpha_i)\in \calA} \alpha_i\pi(T)_{A_i} \ge \sum_{(B_j,\beta_j)\in\calB} \beta_j\pi(T)_{B_j}$ (i.e., its log-projection vector satisfies the linear form of the FNC inequality). 

In fact, the notion of saturation ($\calA$ saturates $\calB$) is a generalization of the notion of covering ($\calA$ covers $\calB$). Put in another way, a NC inequality is automatically an FNC inequality. Intuitively, this is because an one-to-one-mapping between $\Sigma$ and $\Lambda$ can be also viewed as a maximum flow (a flow that saturates every edge going out of $s$) in $\Network(\calA,\calB)$. To be more specific, we state and prove Lemma~\ref{lm:FNC-NC}. Before that, we need some definitions.

Let $K$ be an arbitrary positive integer, and for each $k\in [K]$, let $\phi_k: \sum_{A_i\in \calA^{(k)}}\pi_{A_i}\ge \sum_{B_j\in\calB^{(k)}} \pi_{B_j}$ be a NC inequality (the linear form) and $\gamma_k\ge 0$ be an arbitrary real number. The nonnegative linear combination $\sum_{k=1}^{K}\gamma_{k}\cdot \phi_{k}$, is defined as the following inequality:
\begin{equation}
\label{eq:comb}
\sum_{1\le k\le K}\gamma_k\left(\sum_{A_i\in \calA^{(k)}}\pi_{A_i}\right)\ge \sum_{1\le k\le K}\gamma_{k}\left(\sum_{B_j\in\calB^{(k)}} \pi_{B_j}\right).
\end{equation}

\begin{lemma}
	\label{lm:FNC-NC}
	The set of FNC inequalities (the linear form) is exactly the set of
	all nonnegative linear combinations of NC inequalities (the linear form).
\end{lemma}

\begin{proof}
First we show that a nonnegative linear combination of NC inequalities (the linear form) is a FNC inequality (the linear form). It suffices to show that for arbitrary $K, \gamma_1,\cdots,\gamma_k$ and NC inequalities $\phi_1,\cdots,\phi_k$, the inequality~\ref{eq:comb} satisfies properties C1 and C2, which implies that the nonnegative linear combination is an FNC inequality. 

Property C1 follows from the property P1 in the definition of $\calA$ covers $\calB$. We now show that the inequality~\ref{eq:comb} satisfies property C2.
For each $\phi_{k}$, let $f_k$ be the one-to-one mapping from $\Sigma_k=\{(A_i,t)\mid t\in A_i, A_i\in \calA_k\}$ to $\Lambda_k=\{(B_j,s)\mid s\in B_j, B_j\in \calB_k\}$, whose existence is guaranteed by the property P2 of in the definition of an NC inequality. 
Let family $\calA^*$ be the disjoint union of families $\calA^{(1)},\cdots,\calA^{(K)}$ and family $\calB^*$ be the disjoint union of families $\calB^{(1)},\cdots,\calB^{(K)}$, and let $\Network(\calA^*,\calB^*)$ be the network flow instance in the definition of $\calA$ saturates $\calB$.
We now describe an $s$-$t$ flow in the network $\Network(\calA^*,\calB^*)$ whose value is $\sum_k \gamma_k\left(\sum_j|B_j|\right)$.
For each node labeled by $(A_i, x)$ where $A_i\in \calA^{(k)}$, we send $\gamma_k$ unit of flow along the edge $(s, (A_i, x))$, and we send $\gamma_{k}$ unit of flow along the edge $((A_i, x), (B_j, y))$ where $B_j\in\calB^{(k)}$ and $(B_j, y)=f_{k}((A_i,x))$.
For each node labeled by $(B_j, y)$ where $B_j\in \calB^{(k)}$, we send $\gamma_k$ unit of flow along the edge $((B_j, y),t)$.
It is not hard to verify that the above flow satisfies all capacity constraints and flow conservation constraints, so it is a valid $s$-$t$ flow.
Its value is the total amount of flow that the sink $t$ receives, which is exactly $\sum_k \gamma_k\left(\sum_j|B_j|\right)$ since every node $(B_j, y)$ where $B_j\in \calB^{(k)}$ sends to it $\gamma_k$ unit of flow. Therefore, property C2 is satisfied.

Now we prove that every FNC inequality is a nonnegative linear combination of NC inequalities. Let $\sum_{(A_i,\alpha_i)\in \calA} \alpha_i \pi_{A_i} \ge \sum_{(B_j,\beta_j)\in\calB} \beta_j\pi_{B_j}$ be an FNC inequality (the linear form). For brevity and only in this proof, we rewrite it into the form: $\langle c,\pi \rangle\leq 0$ where $c$ is a real vector. We consider the following two cases.

\textbf{Case 1:} If $c$ is a rational vector (all entries in $c$ are rational numbers), there exists $m\in \mathbb{N}^{+}$ such that $mc$ is integral. Consider the network flow instance associated with the equivalent FNC inequality $\langle mc,\pi \rangle\leq 0$. Since all capacities of the network are integral, there exists an integral maximum flow. By flow-path decomposition, an integral maximum flow can be decomposed into a number of single-path flows, each carrying $1$ unit of flow from $s$ to $t$, passing through two other nodes in between. Viewing this collection of unit-amount-single-path flows as an one-to-one mapping between nodes in $\Sigma$ and nodes in $\Lambda$ (each with integral number of copies), we can see that $\langle mc,\pi \rangle\leq 0$ is indeed an NC inequality, and therefore the equivalent inequality $\langle c,\pi \rangle\leq 0$ is a scaling of a NC inequality.

\textbf{Case 2:} If not all entries of $c$ are rational, we first show that every vector $\pi$ in $\NC$ satisfies $\langle c,\pi \rangle\leq 0$. Suppose not, then there exists a vector $y\in \NC$ such that $\langle c,y\rangle=\epsilon >0$. We claim that there is a sequence of FNC inequalities with rational coefficient vectors $\{\langle c^{(i)},\pi\rangle\leq 0\}_i$ such that $\lim_{i\to+\infty}c^{(i)} = c$. Hence, for some sufficiently large $i$, $\langle c^{(i)},y\rangle \geq \epsilon/2>0$. However, since $c^{(i)}$ is a rational vector, according to the conclusion in Case 1, we know that the inequality $\langle c^{(i)},\pi\rangle\leq 0$ is (a scaling of) an NC inequality, which should be satisfied by any vector $y\in \NC$ according to the definition of $\NC$. This is a contradiction. Now we show existence of the claimed sequence. First of all, the set of vectors $c$ such that $\langle c,\pi \rangle\leq 0$ is a FNC inequality forms a rational polyhedral cone, since this set is characterized by the linear constraints in C1 (whose coefficients are integers) and the max-flow constraint in C2 (which can be captured by linear constraints with integer coefficients using auxiliary flow variables). Let $V$ be the set of rational generating vectors of this cone and let $c= V\gamma$ for some $\gamma\geq 0$ (each column of $V$ is a generating vector). Take a rational nonnegative sequence of vectors $\{\gamma^{(j)}\}_j$ that approaches to $\gamma$, $\{V\gamma^{(j)}\}$ would be the desired sequence.
	
The rest can be seen from Farkas' Lemma: Let $P\pi\leq 0$ be a feasible system of inequalities and $\langle c,\pi\rangle\leq 0$ be an inequality satisfied by all $\pi$ satisfying $P\pi\leq 0$. From Farkas' Lemma, we know that $\langle c,x\rangle\leq 0$ is a nonnegative linear combination of the inequalities in $P\pi\leq 0$ (see e.g., \cite{korte2012combinatorial}).
\end{proof}

\noindent
{\em Proof of Theorem~\ref{thm:nonuniform}}.
It suffices to show that all non-FNC inequalities are false, in the sense that any non-FNC inequality is violated by some object. Consider a non-FNC inequality:
\begin{equation}
\label{eq:nonFNC}
\prod_{(A_i,\alpha_i)\in \calA}x_{A_i}^{\alpha_i}\ge \prod_{(B_j,\beta_j)\in\calB}x_{B_j}^{\beta_j},
\end{equation}
where $\calA$ {\em does not} saturate $\calB$. We construct a rectangular flower $T$ that violates \eqref{eq:nonFNC}.

Consider
the network flow instance $\Network(\calA,\calB)$.
Suppose C1 does not hold:
for some $x\in [n]$,
$\sum_{i=1}^k \alpha_i \cdot\indicator[x\in A_i]\ne \sum_{j=1}^m \beta_j\cdot \indicator[x\in B_j]$.
Assume without loss of generality that $x=1$. We consider two cases. 
First, if 
$\sum_{i=1}^k \alpha_i \cdot\indicator[1\in A_i]>\sum_{j=1}^m \beta_j\cdot \indicator[1\in B_j]$.
Let $T=[0,1/2]\times [0,1]\times\ldots\times [0,1]$.
We can see \eqref{eq:nonFNC} is false for $T$ since $\log \text{LHS}= -\sum_{i=1}^k \alpha_i\cdot \indicator[1\in A_i]$ and
$\log \text{RHS}= -\sum_{j=1}^m \beta_j \cdot\indicator[1\in B_j]$.
Second, if 
$\sum_{i=1}^k \alpha_i \cdot\indicator[1\in A_i]<\sum_{j=1}^m \beta_j\cdot \indicator[1\in B_j]$.
Let $T=[0,2]\times [0,1]\times\ldots\times [0,1]$.
We can see \eqref{eq:nonFNC} is false for $T$ since $\log \text{LHS}= \sum_{i=1}^k \alpha_i\cdot \indicator[1\in A_i]$ and
$\log \text{RHS}=\sum_{j=1}^m \beta_j \cdot\indicator[1\in B_j]$.

Now suppose C2 is false, that is, the value of the minimum $s$-$t$
cut of $\Network(\calA,\calB)$ is strictly less than $\sum_j \beta_j|B_j|$.
Suppose the minimum $s$-$t$ cut defines the partition $(\calS,\calT)$ of vertices such that $s\in \calS$
and $t\in \calT$ (see Figure~\ref{fig:flower}(ii)).
Let $\Sigma$ and $\Lambda$ be defined as above, and
$\Sigma_\calS=\Sigma\cap \calS$,
$\Sigma_\calT=\Sigma\cap \calT$,
$\Lambda_\calS=\Lambda\cap \calS$,
$\Lambda_\calT=\Lambda\cap \calT$.
Since the min-cut is strictly less than $\sum_j \beta_j|B_j|$, none of the above four sets are empty.
Clearly, there is no edge from $\Sigma_S$ to $\Lambda_T$ since otherwise
the size of the cut would be $+\infty$.
In other words, $\Lambda_\calS$ absorbs all outgoing edges from $\Sigma_\calS$.

Moreover, we can see the value of the min-cut is $\sum_{(A_i, x)\in \Sigma_\calT}\alpha_i+\sum_{(B_j, y)\in \Lambda_\calS}\beta_j$.
Since this value is less than $\sum_{(B_j, y)\in \Lambda} \beta_j$,
we have $\sum_{(A_i, x)\in \Sigma_\calT}\alpha_i< \sum_{(B_j, y)\in \Lambda_\calT}\beta_j$.
Now, we construct the rectangular flower $T$.
Suppose
$T=\bigcup_{S\subseteq [n], S\ne \emptyset} T_S$ with each $T_S$ being a closed rectangle in $\span(S)$, and for each $x\in S$ we use $T_{S,x}$ to denote the edge length of $T_S$ along axis $x$.
We specify all $T_{S,x}$ as follows:
$$
T_{S,x}=\left\{
\begin{array}{rcl}
2       &      & {\text{if there exists a }B_j\text{ such that }S\subseteq B_j, \text{ and } (B_j, x)\in \Lambda_\calT},\\
1       &      & \text{otherwise.}
\end{array} \right. $$

We now verify that the above rectangular flower $T$ violates the given non-FNC inequality.
By definition of $T_{S,x}$ above, we can show that $T_{A_i,x}=1$ for any node
$(A_i,x)\in \Sigma_\calS$. This is because there does not exist a $B_j$ such that $A_i\subseteq B_j$ and $(B_j,x)\in \Lambda_\calT$, otherwise, there will be an edge from $\Sigma_\calS$ to $\Lambda_\calT$ (according to the definition of the edges in $\Network(\calA,\calB)$), leading to a contradiction. Note that $T_{A_i,x}\le 2$ for any node $(A_i,x)\in \Sigma_\calT$, we have 
\[
\log\left(\prod_{(A_i,\alpha_i)\in \calA}|T_{A_i}|^{\alpha_i}\right)
\le \sum_{(A_i, x)\in \Sigma_\calT}\alpha_i.
\]
On the other hand, we have
\[
\log\left(\prod_{(B_j,\beta_j)\in \calB}|T_{B_j}|^{\beta_j}\right)
\ge \sum_{(B_j, y)\in \Lambda_\calT}\beta_j,
\]
which implies the given inequality is false for $T$.
This completes the proof for Theorem~\ref{thm:nonuniform}.
\qed

We denote the set of log-projection vectors generated by rectangular flowers as
$$
\RF =\{ \pi \in \R^{2^n-1} \mid \pi \text{ is the log-projection vector of a rectangular flower } T \subseteq \mathbb{R}^n\}.
$$
Now, we prove Theorem~\ref{thm:recflower}.


\vspace{0.3cm}
\noindent
{\em Proof of Theorem~\ref{thm:recflower}}.
By definition, $\RF\subseteq \Con$. If suffices to show that $\RF=\NC$.
Note that a rectangular flower $T$ can be characterized by the edge lengths of its orthogonal projections along all axes, i.e., $\{T_{S,x}\}_{x\in S}$, a given vector $\pi$ is the log-projection vector of some rectangular flower in $\R^n$ iff the following linear program, denoted as $\LP(\pi)$, is feasible (treating $f_{S,x}$ as variables\footnote{It is intended that for all $x\in [n], S\subseteq [n]$ such that $i\in S$, $f_{S,x}=\log T_{S,x}$.}):
\begin{eqnarray*}
	\sum_{x\in S}f_{S,x} &=& \pi_S, \quad \text{ for all }S\subseteq [n], \\
	f_{S,x} &\geq & f_{S',x}, \quad \text{ for all }x\in S\subseteq S' \subseteq [n].
\end{eqnarray*}
With this notation, we have $\RF =\{ \pi \in \R^{2^n-1} \mid \LP(\pi) \text{ is feasible}\}$.
It is not hard to check that $\RF$ is a convex cone
(i.e., if $\pi_1, \pi_2 \in \RF$, $a\pi_1+b\pi_2\in \RF$ for any $a,b>0$).
In fact, $\RF$ is a closed polyhedral cone, as it is the intersection of finitely many closed halfspaces.

We first show that each point in $\RF$ satisfies all NC inequalities. Let $T\subseteq \R^n$ be a rectangular flower and consider an arbitrary NC inequality (linear form)
$\sum_{A_i\in \calA}\pi_{A_i}\ge \sum_{B_j\in\calB}\pi_{B_j}$.
By definition of a NC inequality,
there exists an one-to-one mapping $f$ that maps $(A_i,t) (t\in A_i)$ to $(B_j,s)=f(A_i,t)$, such that $t=s$ and $A_i\subseteq B_j$.
By the definition of a rectangular flower, 
for any $x\in S\subseteq S' \subseteq [n]$, we have $\log T_{S,x} \geq \log T_{S',x}$. Let $\pi(T)_{S,x}=\log T_{S,x}$, we have
\[\sum_{A_i\in \calA}\pi(T)_{A_i}=\sum_{A_i\in \calA}\sum_{x\in A_i}\pi(T)_{A_i,x}\ge \sum_{A_i\in \calA}\sum_{x\in A_i}\pi(T)_{f(A_i,x)}=\sum_{B_i\in \calB}\sum_{x\in B_j}\pi(T)_{B_j,x}=\sum_{B_j\in\calB}\pi(T)_{B_j},\]
where the second equality holds because $f$ is an one-to-one mapping. This proves that $\RF\subseteq \NC$.

We next show that $\NC\subseteq \RF$.
Suppose for contradiction that there is a point $\pi^*\in \NC$ such that $\pi^* \not\in \RF$.
Since $\RF$ is a closed polyhedral cone, there is a hyperplane $\sum_{S\subseteq [n]} \alpha_S \pi_{S}= 0$ separating
$\RF$ and $\pi^*$ (with $\sum_{S\subseteq [n]} \alpha_S \pi^*_{S}> 0$).
So $\sum_{S\subseteq [n]} \alpha_S \pi_{S} \leq  0$  is not an FNC inequality (since $\pi^*\in \NC$ should satisfy all FNC inequalities).
From the proof Theorem~\ref{thm:nonuniform}, we have shown that for any non-FNC inequality, there exists a rectangular flower that violates it.
This contradicts to the fact that $\sum_{S\subseteq [n]} \alpha_S \pi_{S}\leq 0$ for all $\pi\in \RF$.
Hence, $\NC\subseteq \RF$.
This concludes the proof for Theorem \ref{thm:recflower}.
\qed

At the end of this section, we briefly discuss projection inequalities with nonzero constant terms, e.g., $\sum_S \alpha_S \pi_S \leq \beta$, for $\beta\ne 0$.
If $\beta<0$, none of such inequality is true by considering the unit hypercube.
If $\beta> 0$, whether or not the inequality is true entirely depends on the correctness of its homogeneous counterpart, the inequality $\sum_S \alpha_S \pi_S \leq 0$.
To see that, on one hand, if $\sum_S \alpha_S \pi_S\leq 0$ is true for all $\pi\in \Con$, 
so is $\sum_S \alpha_S \pi_S\leq \beta$ for any $\beta>0$.
On the other hand, if $\sum_S \alpha_S \pi_S\leq \beta$ is true for some $\beta>0$, $\sum_S \alpha_S \pi_S\leq 0$ must also be true (we will prove this fact in Section~\ref{sec:nonconvex}). Therefore, it suffices to consider only those inequalities with zero constant term.

\section{Proof of Theorem~\ref{thm:nonconvex}: Non-Convexity of $\Con$}
\label{sec:nonconvex}

In this section we prove Theorem~\ref{thm:nonconvex}. We prove by contradiction that the constructible region $\Con$ for $n\geq4$ is non-convex. 
Before that, we need some definitions.
A hyperplane of the form $\langle c,\pi\rangle= b$ in $\R^{2^n-1}$ is called a \emph{supporting hyperplane} of $\Psi_n$ if it has non-empty intersection with the closure of $\Psi_n$ and all points in $\Psi_n$ lie on the same side of the hyperplane (namely, either all $\pi\in \Psi_n$ satisfy $\langle c,\pi\rangle\le b$ or all $\pi\in \Psi_n$ satisfy $\langle c,\pi\rangle\ge b$).
We first show that each supporting hyperplane of $\Con$ must pass through the origin.


\begin{proposition}\label{prop:hyperplane}
For any positive integer $n$, let $\sum_i\alpha_i\pi_{A_i}-\sum_j \beta_j\pi_{B_j}+C=0$ be a supporting hyperplane of $\Con$ such that $\alpha_{i}, \beta_{j}\geq 0$, then $C=0$.
\end{proposition}

\begin{proof}
Assume without loss of generality that all $\pi\in \Psi_n$ satisfy $\sum_i\alpha_i\pi_{A_i}-\sum_j \beta_j\pi_{B_j}+C\ge 0$. In this proof, for an object $T$ and a vector $v\in \R^n$, we define $T+v=\{x+v\mid x\in T\}$.

Consider the $n$-dimensional unit hypercube $C_n$ in $\mathbb{R}^{n}$. It can be seen that its projection onto any subspace has volume $1$, and therefore $\pi(C_n)_S=0$ for all $S\subseteq [n]$, which implies that $C\geq 0$. Via a similar argument in the proof of Theorem~\ref{thm:nonuniform}, we can show that $\sum_{i} \alpha_i \cdot \indicator[x\in A_i]=\sum_{j} \beta_j \cdot\indicator[x\in B_j]$ holds for all $x\in [n]$, since otherwise we can scale the unit hypercube along one axis by a factor of $2$ or $1/2$ (depending on which side of the above equality is larger) to disprove the inequality.
	
Assume the contrast that $C>0$. 
Since $\sum_i\alpha_i\pi_{A_i}-\sum_j \beta_j\pi_{B_j}+C=0$ is a supporting hyperplane of $\Con$, we know that there exists an object $T$ such that $\sum_i\alpha_i\pi(T)_{A_i}-\sum_j \beta_j\pi(T)_{B_j}\leq -0.99C$ (since by definition the hyperplane has non-empty intersection with the closure of $\Con$). Denote $\Omega(T)=\sum_i\alpha_i\pi(T)_{A_i}-\sum_j \beta_j\pi(T)_{B_j}$. We first show that there exists another object $T'$ such that $\Omega(T')\leq -2C/3$. Moreover, $T'$ can be represented as the union of a set of unit hypercubes in $\mathbb{R}^{n}$.
	
For any $\varepsilon>0$ and an $n$-dimensional integral vector $v=(v_1,\cdots,v_n)$, we define $C^{\epsilon}_{v}$ to be the hypercube 
$C^{\epsilon}_{v}=[v_1\epsilon,(v_1+1)\epsilon]\times[v_2\epsilon,(v_2+1)\epsilon]\times\cdots\times[v_n\epsilon,(v_n+1)\epsilon]$ of size $\epsilon$.
And we define 
$$T(\varepsilon)=\bigcup_{v : C^{\epsilon}_{v}\cap T\ne\emptyset}C^{\epsilon}_{v}.$$
Intuitively, $T(\varepsilon)$ is a ``union-of-hypercube $\epsilon$-approximation'' of $T$.
It is not hard to see that for all $S\subseteq [n]$, we have $\lim_{\varepsilon\rightarrow 0}\pi(T(\varepsilon))_{S}=\pi(T)_{S}$. Therefore, there exists some small enough $\epsilon^*>0$ such that $\Omega(T(\epsilon^*))\leq -2C/3$. 
Now let $T^{'}=\frac{1}{\epsilon^*}T(\epsilon^*)=\{\frac{1}{\epsilon^*}\cdot x\mid x\in T(\epsilon^*)\}$. It can be seen that $T'$ is the union of a set of unit hypercubes in $\mathbb{R}^{n}$. We write $T'=\bigcup_{v\in I}C^1_v$ where $I$ is a set of $n$-dimensional integral vectors. Moreover, we have that
\begin{equation}
\begin{aligned}
\Omega(T')=&\sum_i\alpha_i\pi(T')_{A_i}-\sum_j \beta_j\pi(T')_{B_j}\\
=&\sum_i\alpha_i\pi(T(\epsilon^*))_{A_i}+\log\frac{1}{\epsilon^*}\cdot\left(\sum_{i} \alpha_i |A_i|\right)-\sum_j \beta_j\pi(T(\epsilon^*))_{B_j}-\log\frac{1}{\epsilon^*}\cdot\left(\sum_{j} \beta_j |B_j|\right)\\
= & \text{ }\Omega(T(\delta))\leq -2C/3,
\end{aligned}
\end{equation}
where the last equality comes from $\sum_{i} \alpha_i \cdot \indicator[x\in A_i]=\sum_{j} \beta_j \cdot\indicator[x\in B_j]$ for all $x\in [n]$.

We now construct an object $T''$ such that $\pi(T'')_{S}=2\pi(T')_{S}$ for all $S\subseteq [n]$. Note that this implies $\Omega(T'')=2\cdot\Omega(T')\leq -4C/3<-C$, which is a contradiction to the assumption that all constructible vectors $\pi$ satisfy $\sum_i\alpha_i\pi_{A_i}-\sum_j \beta_j\pi_{B_j}+C\ge 0$. Let $m$ be a sufficiently large integer such that $T'\subseteq [0,m]^{n}$. Now we define
\[T''=\bigcup_{v\in I} (T'+m\cdot v).\]
Intuitively, $T''$ is the union of $|T'|$ copies of $T'$ that are placed ``in the same shape'' of $T'$. See Figure~\ref{fig:double} for an example.
It is not hard to see that, for each $S\subseteq [n]$, $T''_{S}$ consists of $|T'_{S}|$ disjoint copies of $T'_{S}$ that are placed ``in the same shape'' of $T'_S$. Therefore, for each $S\subseteq [n]$, we have that $|T''_{S}|=|T'_{S}|\cdot |T'_{S}|$ and $\pi(T'')_{S}=2\pi(T')_{S}$. This completes the description of the construction of $T''$, and also completes the proof by contradiction of the proposition.
\end{proof}

\begin{figure}[h]
	\centering
\begin{tikzpicture}

\draw[black, very thick] (0,0)--(0,1)--(2,1)--(2,0)--(0,0);
\draw[black, very thick] (0,0)--(1,0)--(1,2)--(0,2)--(0,0);

\node at (1,-0.5) {$T'$};

\draw[black, very thick] (5,0)--(5,1)--(7,1)--(7,0)--(5,0);
\draw[black, very thick] (5,0)--(6,0)--(6,2)--(5,2)--(5,0);

\draw[black, very thick] (7,0)--(7,1)--(9,1)--(9,0)--(7,0);
\draw[black, very thick] (7,0)--(8,0)--(8,2)--(7,2)--(7,0);

\draw[black, very thick] (5,2)--(5,3)--(7,3)--(7,2)--(5,2);
\draw[black, very thick] (5,2)--(6,2)--(6,4)--(5,4)--(5,2);
\draw [ultra thick, draw=black, fill=gray, opacity=1]
(6,1) -- (6,2) -- (7,2) -- ((7,1) -- cycle;
\node at (7,-0.5) {$T''$};
\end{tikzpicture}
\caption{(i) An object $T'\subseteq \R^2$ consists of $3$ unit squares.
(ii) The object $T''\subseteq\R^2$ defined according to $T'$ consists of $9$ unit squares. There is no square in the dark area.}
\label{fig:double}
\end{figure}
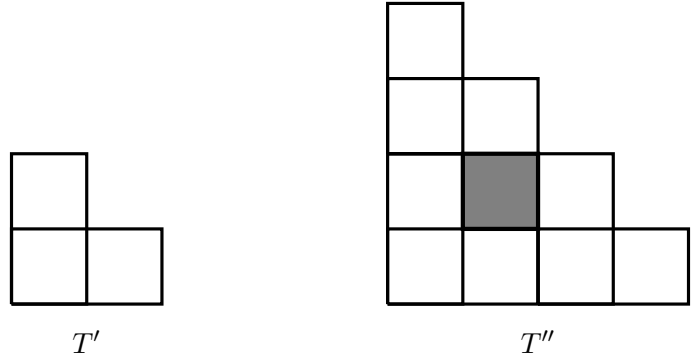

Now, for $n\ge 4$, consider the projection of $\Con$ onto the subspace spanned by coordinate vectors $e_{\{1,2\}}$, $e_{\{1,3\}}$, $e_{\{2,3\}}$, $e_{\{2,4\}}$, $e_{\{3,4\}}$, $e_{\{1,2,3\}}$ and $e_{\{2,3,4\}}$:
\[\Pi_{n}=\bigg\{\big(\pi(T)_{\{1,2\}},\pi(T)_{\{1,3\}},\pi(T)_{\{2,3\}},\pi(T)_{\{2,4\}},\pi(T)_{\{3,4\}},\pi(T)_{\{1,2,3\}},
\pi(T)_{\{2,3,4\}}\big)\text{ }\bigg|\text{ } T\subseteq \mathbb{R}^{n}\bigg\}.\]

Let $b$ be a real number and $v\in \R^7$ be a real vector, whose coordinates are indexed by sets in ${\{1,2\}}$, ${\{1,3\}}$, ${\{2,3\}}$, ${\{2,4\}}$,${\{3,4\}}$, ${\{1,2,3\}}$, and ${\{2,3,4\}}$ respectively. Define $\tilde{v}\in \R^{2^n-1}$ to be the extension of the $7$-dimensional vector $v$, such that $\tilde{v}_I=v_I$ for $I$ being each one of the above $7$ sets indexing the coordinates of $v$, and $\tilde{v}_I=0$ for other $I$.
It is not hard to see that for each supporting hyperplane $\langle v,\omega \rangle=b$ of $\Pi_{n}$, its ``extension'' $\langle \tilde{v},\pi \rangle=b$ must be a supporting hyperplane of $\Con$. Combine this with Proposition~\ref{prop:hyperplane}, we know that such a supporting hyperplane of $\Pi_{n}$ must pass through the origin as well, namely $b=0$. Now assume without loss of generality that $\langle \tilde{v},\pi \rangle\le 0$ is satisfied by all $\pi\in \Con$, we know from Theorem~\ref{thm:nonuniform} that the inequality $\langle \tilde{v},\pi \rangle\le 0$ must be an FNC inequality. Additionally, in this FNC inequality the only terms with non-zero coefficients are $\pi_{\{1,2\}}$, $\pi_{\{1,3\}}$, $\pi_{\{2,3\}}$, $\pi_{\{2,4\}}$, $\pi_{\{3,4\}}$, $\pi_{\{1,2,3\}}$ and $\pi_{\{2,3,4\}}$.

Now we show that any FNC inequality involving only the above terms 
is a nonnegative linear combination of the following BT inequalities (the linear form):
\begin{eqnarray}
\pi_{\{1,2\}}+\pi_{\{1,3\}}+\pi_{\{2,3\}} \geq 2\cdot\pi_{\{1,2,3\}},\label{bt:1}\\
\pi_{\{2,3\}}+\pi_{\{3,4\}}+\pi_{\{2,4\}} \geq 2\cdot\pi_{\{2,3,4\}}.\label{bt:2}
\end{eqnarray}

By Lemma~\ref{lm:FNC-NC}, it suffices to consider NC inequalities with only terms with non-zero coefficients being $\pi_{\{1,2\}}$, $\pi_{\{1,3\}}$, $\pi_{\{2,3\}}$, $\pi_{\{2,4\}}$, $\pi_{\{3,4\}}$, $\pi_{\{1,2,3\}}$, and $\pi_{\{2,3,4\}}$.
It is not hard to see from property P2 that the RHS of such an NC inequality may only contain terms $\pi_{\{1,2,3\}}$ and $\pi_{\{2,3,4\}}$. According to Corollary~\ref{cor:BTcomb} of the Exact Single Cover Theorem (which will be proved in the next section), such NC inequality must be a nonnegative combination of \eqref{bt:1} and \eqref{bt:2}. If $\Con$ is convex, then $\Pi_{n}$ is also convex, and therefore $\Pi_{n}$ is characterized by \eqref{bt:1} and \eqref{bt:2}. 
However, we show that this is not true by analyzing a variation of an example first proposed in \cite{bollobas1995projections}. Specifically, we consider the vector $\omega=(0,2,0,2,0,1,1)$. It is easy to verify that $\omega$ satisfies \eqref{bt:1} and \eqref{bt:2}. We will show that $\omega\notin \Pi_{n}$. Before that, we need some definitions.
Let $\calA$ be a $k$-cover of $[n]$, then $\calA$ defines a partition of $[n]$ into equivalence classes in the way that two elements being in the same equivalence class iff they lie in exactly the same subsets in $\calA$.
We state and use the following Theorem (Theorem 4 from \cite{bollobas1995projections}, restated here). 
\begin{theorem}\label{thm:4in3}
Let $T$ be an object in $\mathbb{R}^{n}$, and let $\calA$ be a $k$-cover of $[n]$ such that $\prod_{A\in \calA}|T_{A}|=|T|^{k}$.
Then $T=\prod T_{E}$, the product being over all equivalence classes of the cover $\calA$.
\end{theorem}
Suppose there exists an object $T$ whose log-projection vector is consistent with $\omega$ on the coordinates indexed by the sets ${\{1,2\}}$, ${\{1,3\}}$, ${\{2,3\}}$, ${\{2,4\}}$, ${\{3,4\}}$, ${\{1,2,3\}}$ and ${\{2,3,4\}}$.
That is,
$|T_{\{1,2\}}|=|T_{\{2,3\}}|=|T_{\{3,4\}}|=1$, $|T_{\{1,3\}}|=|T_{\{2,4\}}|=4$, and $|T_{\{1,2,3\}}|=|T_{\{2,3,4\}}|=2$.
It can be seen that $|T_{\{1,2\}}|\cdot |T_{\{1,3\}}|\cdot |T_{\{2,3\}}|=|T_{\{1,2,3\}}|^{2}$. From Theorem~\ref{thm:4in3}, we know that $T_{\{2,3\}}$ must be the Cartesian product of its one-dimensional projections where $|T_{\{2\}}|=1/2$ and $|T_{\{3\}}|=2$. Similarly, we also have $|T_{\{2,3\}}|\cdot |T_{\{2,4\}}|\cdot |T_{\{3,4\}}|=|T_{\{2,3,4\}}|^{2}$. Therefore, according to Theorem~\ref{thm:4in3}, we know that $T_{\{2,3\}}$ must be the Cartesian product of its one-dimensional projections where $|T_{\{2\}}|=2$ and $|T_{\{3\}}|=1/2$. This causes a contradiction. Therefore, $\omega\notin \Pi_{n}$ and $\Pi_{n}$ is not characterized by \eqref{bt:1} and \eqref{bt:2}. As a result, $\Con$ is not convex and this completes the proof of Theorem~\ref{thm:nonconvex}.

\section{Counterexample Construction for NC$\setminus$BT Inequalities}
\label{sec:counterex}

We have shown that the constructible region $\Con$ cannot be fully characterized by a set of linear inequalities as it is non-convex.
However, it is still interesting to see what is the set of all correct linear inequalities (inequalities that are true for all $\pi\in \Con$).
Equivalently, we want to figure out the set of linear inequalities that define $\Conv(\Con)$, the convex hull of $\Con$.

In this section, we construct counterexamples for several NC but non-BT (denoted as NC$\setminus$BT) inequalities.
Note that a compact object can be arbitrarily approximated by the union of unit hypercubes, we consider such objects in
our counterexamples. In this subsection, we use a $n$-tuple $\textbf{t}=(t_1,t_2,\ldots,t_n)$ where $\{t_i\}_{i\in [n]}$ are non-negative integers to represent the
$n$-dimensional unit hypercube $[t_1,t_1+1]\times[t_2,t_2+1]\times\cdots\times[t_n,t_n+1]$.


\subsection{Skeleton}

We need the notion of a \emph{skeleton}, which is central to our counterexample construction.

\begin{definition}
Let $G=(V,E)$ be an undirected graph where $V=\{v_1,\cdots,v_n\}$ and let $C\subseteq V, C\ne\emptyset$ be such that the induced subgraph $G[C]$ is a clique. Let $M$ be a positive integer. We define $\SK_{M}(G,C)$, the \emph{skeleton} of $C$ on $G$ with parameter $M$, as the union of all unit hypercubes $\textbf{t}$ such that $\forall i\in C, 0\le t_i\le M-1$ and $\forall i\notin C, t_i=0$.
Let $C_1,C_2,\cdots,C_s\subseteq V$ be all non-empty maximal (inclusion-wise) cliques in $G$.
The \emph{skeleton} of $C$ in $G$ with parameter $M$ is defined as
\[
\SK_{M}(G)=\bigcup_{r=1}^{s} \SK_{M}(G,C_r).
\]
See Figure~\ref{fig:flower2} for an example.
\end{definition}

Consider an arbitrary FNC inequality $\prod_{\calA}x^{\alpha_i}_{A_i}\ge \prod_{\calB}x^{\beta_j}_{B_j}$. In order to disprove this inequality, we construct a skeleton with large RHS value and small LHS value. To this end, we need the following definition for a FNC inequality.

\begin{definition}
We define the \emph{connection graph} for the above inequality to be
an undirected graph $G=(V,E)$ where $V=\{v_1,\cdots,v_n\}$, 
and the edge $(v_x,v_y)\in E$ iff $x$ and $y$ appear simultaneously in some $B_j$ but not in any $A_i$.
\end{definition}

\begin{figure}[h]
	\centering
	\includegraphics[width=0.7\linewidth]{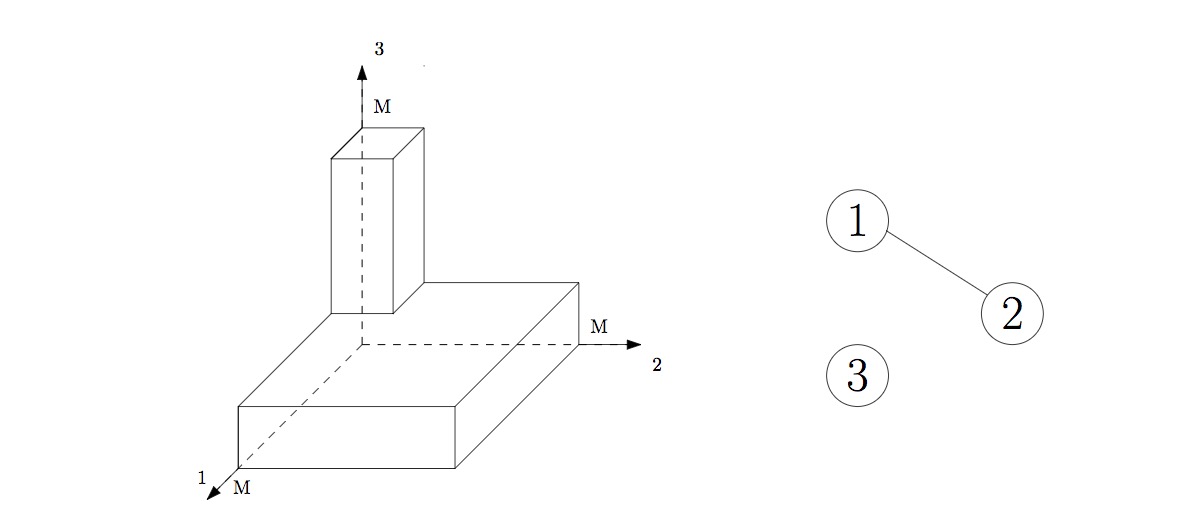}
	\caption{(i) Skeleton.
		(ii) Graph $G$.}
	\label{fig:flower2}
\end{figure}

Let graph $G=(V,E)$ be the connection graph for the FNC inequality $\prod_{\calA}x^{\alpha_i}_{A_i}\ge \prod_{\calB}x^{\beta_j}_{B_j}$. We analyze whether the skeleton $\SK_{M}(G)$ defined above is a counterexample for the FNC inequality, for sufficiently large $M$. For a subset of vertices $S\subseteq V$,
let $\omega(S)$ be the size of the maximum clique in the induced subgraph $G[S]$.
As $M$ goes to $+\infty$,
we have the following asymptotic estimations:
\[
\prod_{(A_i,\alpha_i)\in \calA}|\SK_{M}(G)_{A_i}|^{\alpha_i}=\Theta\left(M^{\sum_{i=1}^{k}\alpha_i}\right);
\text{   }\prod_{(B_j,\beta_j)\in\calB}|\SK_{M}(G)_{B_j}|^{\beta_j}=\Theta\left(M^{\sum_{j=1}^{m}\beta_j\cdot\omega(B_j)}\right).\]

The following lemma is a direct consequence of the above estimations.
\begin{lemma}
If the following inequality holds for a FNC inequality
$\prod_{\calA}x^{\alpha_i}_{A_i}\ge \prod_{\calB}x^{\beta_j}_{B_j}$,
\[
\sum_{(A_i,\alpha_i)\in \calA}\alpha_i<\sum_{(B_j,\beta_j)\in\calB}\beta_j\cdot\omega(B_j).
\]
the FNC inequality is incorrect, i.e., there exists a counterexample to it.
\end{lemma}

\begin{example}
\label{ex:1}
Consider the NC inequality
$x_{\{1,2\}}\cdot x_{\{2,3\}}\cdot x_{\{3,4\}} \geq x_{\{1,2,3\}}\cdot x_{\{2,3,4\}}$.
The edge set of its connection graph $G$ is $\{(1,3), (2,4)\}$.
	We have $\sum_i\alpha_i=3$ and $\sum_j \beta_j \cdot\omega(B_j)=4$.
	Hence, the inequality is not true in general.
\end{example}


\subsection{Union of Boxes}
Let $\mathbf{b}$ be an $n$-dimensional vector with all coordinates being positive numbers.
We define a box $B(\mathbf{b})$ to be a hypercube $B(\textbf{b})=\{\mathbf{x}\mid 0\le x_i\le b_i\}$.
In this subsection we denote the sum of two sets by their Minkowski sum, namely $A+B=\{a+b\mid a\in A, b\in B\}$. 
For example, let $\mathbf{1}$ be the all-one vector, then $B(\mathbf{b})+\mathbf{1}=\{\mathbf{x}\mid 1\le x_i\le b_i+1\}$.

The objects in this subsection are
the disjoint union of two boxes $B_1$ and $B_2$.
Here we require not only $B_1$ and $B_2$ are disjoint in $\mathbb{R}^n_{+}$,
but for any $S\subseteq [n]$, their projections onto subspace $\mathbb{R}^S$ are disjoint as well.
In particular, we let $T$ be the union of the following two boxes:
\[
B_1=B(\textbf{1});\text{ }B_2=B(M^{t_1},M^{t_2},\cdots, M^{t_n})+\textbf{1},
\]
where $t_i\in \mathbb{R}$ and $M>0$. 

Consider the FNC inequality $\prod_{(A_i,\alpha_i)\in \calA}x^{\alpha_i}_{A_i}\ge \prod_{(B_j,\beta_j)\in \calB}x^{\beta_j}_{B_j}$.  As $M$ goes to $+\infty$,
we have the following asymptotic estimations:
\[
\prod_{(A_i,\alpha_i)\in \calA}|T_{A_i}|^{\alpha_i}=\Theta\left(M^{\sum_{i=1}^{k}\alpha_i \cdot\max\{0,\sum_{s\in A_i} t_s\}}\right),
\text{ }
\prod_{(B_j,\beta_j)\in \calB}|T_{B_j}|^{\beta_j}=\Theta\left(M^{\sum_{j=1}^{m}\beta_j \cdot\max\{0,\sum_{s\in B_j} t_s\}}\right),
\]
and the following lemma.

\begin{lemma}
\label{lm:unionbox}
If there exist real numbers $t_1,\cdots,t_n$ such that the following inequality holds for a FNC inequality
$\prod_{\calA}x^{\alpha_i}_{A_i}\ge \prod_{\calB}x^{\beta_j}_{B_j}$,
\[
\sum_{(A_i,\alpha_i)\in \calA}\alpha_i \cdot\bigg|\sum_{s\in A_i} t_s\bigg|<\sum_{(B_j,\beta_j)\in \calB}\beta_j\cdot\bigg|\sum_{s\in B_j} t_s\bigg|.
\]
then the FNC inequality is incorrect.
\end{lemma}
\begin{proof}
Our counterexample is $T$ as stated before, which is the union of two boxes
$B_1=B(\mathbf{1}),B_2=B(M^{t_1},M^{t_2},\cdots, M^{t_n})+\mathbf{1}$.
Since for any real number $a$, we have $2\cdot\max\{a,0\}=a+|a|$.
According to the above asymptotic estimations and the following equation that follows from C1 in the definition of a FNC inequality: 
\[\sum_{(A_i,\alpha_i)\in \calA}\sum_{s\in A_i} \alpha_i\cdot t_s=\sum_{(B_j,\beta_j)\in \calB}\sum_{s\in B_j} \beta_j\cdot t_s,\] 
we conclude  that
\[
\sum_{(A_i,\alpha_i)\in \calA}\alpha_i \cdot\max\bigg\{0,\sum_{s\in A_i} t_s\bigg\} <
\sum_{(B_j,\beta_j)\in \calB}\beta_j \cdot\max\bigg\{0,\sum_{s\in B_j} t_s\bigg\}.
\]
Hence,
the object $T$ is a counterexample.
\end{proof}

\begin{example}
Again, consider the NC inequality in Example~\ref{ex:1}.
Let $(t_1,t_2,t_3,t_4)=(1,-1,1,-1)$.
We can see the condition of Lemma~\ref{lm:unionbox} is satisfied and therefore the inequality is incorrect.
\end{example}

\begin{example}
\label{ex:2}
Consider the NC inequality
$x_{\{1,3\}}\cdot x_{\{2,3\}}\cdot x_{\{1,2,4\}} \geq x_{\{1,2,3\}}\cdot x_{\{1,2,3,4\}}$. 
Let $(t_1,t_2,t_3,t_4)=(-1,-1,1,2)$.
The condition of Lemma~\ref{lm:unionbox} is satisfied and therefore the inequality is incorrect.
\end{example}


\subsection{Exact Single Cover Theorem}

Recall that the only class of correct inequalities we know so far, the class of BT inequalities (e.g.,
$|T|^k\le \prod_{A\in \calA} |T_A|$), satisfy that the family of index subsets $\calA$ forms a uniform-cover of $[n]$.
Although we are unable to prove that all correct inequalities are BT inequalities or its nonnegative combinations, we prove the following weaker theorem
which is a necessary condition for an FNC inequality to be true for all constructible vectors, using the union of boxes method.
Intuitively, it says that any single set $B_j\in \calB$ in a correct FNC inequality can be ``uniformly covered'' by a ``sub-family'' of $\calA$.

In this subsection, we let $\mathbf{a}_i\in\{0,1\}^n$ be the 0/1 indicator vector for set $A_i$ and $\mathbf{b}_j\in\{0,1\}^n$ for $B_j$, i.e.,
$\mathbf{a}_{it}=1$ if and only if $t\in A_i$.

\begin{theorem}(Exact Single Cover Theorem)
\label{thm:singlecover}
If the FNC inequality 
$\prod_{\calA}x^{\alpha_i}_{A_i}\ge \prod_{\calB}x^{\beta_j}_{B_j}$
is true for all constructible vectors,
then for each $(B_j,\beta_j)\in \calB$, there exist real numbers $c_{1},c_{2},\cdots,c_{|\calA|}$ such that $0\le c_i\le \alpha_i$ for all $i$ and 
	$$
	\sum_{(A_i,\alpha_i)\in \calA}c_{i} \mathbf a_{i}=\beta_{j} \mathbf b_j.
	$$	
\end{theorem}

\begin{proof}
	
Let $K=\{\sum_{(A_i,\alpha_i)\in \calA}c_{i} \mathbf{a}_i\mid 0\leq c_{i}\leq \alpha_{i}, \forall i\in [k]\}$.
It is not hard to see that $K$ is a closed convex subset of $\mathbb{R}^{n}$.
Assume the contrary that $\beta_{j}\textbf b_j\notin K$, then by the separating hyperplane theorem, we know that there exists a vector $\mathbf{t}\in \mathbb{R}^n$ and a real number $a$ such that
\[
\langle\mathbf {t},\mathbf{x}\rangle < a, \forall \textbf{x}\in K \quad \text{but}\quad
\beta_{j}\cdot\langle\mathbf {t},\mathbf{b}_j\rangle >a.
\]
Let object $T$ be the union of the following two boxes:
	\[
	B_1=B(\textbf{1});\text{  }B_2=B(M^{t_1},M^{t_2},\cdots, M^{t_n})+\textbf{1}.
	\]
	We can see that
	\[
	\prod_{(B_j,\beta_j)\in \calB}|T_{B_j}|^{\beta_j} \geq M^{\beta_{j}\cdot\langle\mathbf {t},\mathbf{b}_j\rangle}>M^a.
	\]
On the other hand, using the asymptotic estimations in the last subsection, as $M$ goes to $+\infty$, we have that
	\[
	\prod_{(A_i,\alpha_i)\in \calA}|T_{A_i}|^{\alpha_i}=\Theta\left( M^{\sum_{(A_i,\alpha_i)\in \calA}\alpha_i \cdot\max\{0,\sum_{s\in A_i} t_s\}}\right)
	=\Theta\left( M^{\sum_{i:\langle\mathbf{a}_i, \mathbf{t}\rangle\ge 0}\alpha_{i} \cdot\langle\mathbf{a}_{i}, \mathbf{t}\rangle}\right)
	<M^{a},
	\]
where the last inequality holds since
$\sum_{i:\langle\mathbf{a}_i, \mathbf{t}\rangle\ge 0}\alpha_{i} \mathbf{a}_{i}\in K$. However, by definition we know that $\langle\mathbf {t},\mathbf{x}\rangle < a, \forall \textbf{x}\in K$.
This contradicts the assumption that $\beta_{j}\textbf b_j\notin K$ and finishes the proof of Theorem~\ref{thm:singlecover}.
\end{proof}

We now present two simple corollaries of Theorem~\ref{thm:singlecover}.

\begin{corollary}
Suppose the FNC inequality
$\prod_{\calA}x^{\alpha_i}_{A_i}\ge \prod_{\calB}x^{\beta_j}_{B_j}$
is true for all constructible vectors,
and the indicator vectors $\{\mathbf{a}_i\}_{A_i\in\calA}$ are linearly independent.
Then this inequality is a
nonnegative combination of at most $|\calB|$ BT inequalities.
\end{corollary}
\begin{proof}
Let $\mathbf{A}$ ($\mathbf{B}$ resp.) be
the matrix with $\mathbf{a}_i$ being the $i$th column ($\mathbf{b}_j$ the $j$th column).
Let $\mathbf{\alpha}=(\alpha_1,\cdots, \alpha_{|\calA|})^T$ and
$\mathbf{\beta}=(\beta_1,\cdots, \beta_{|\calB|})^T$.
By the definition of a FNC inequality, we know that
$
\mathbf{A\alpha}=\mathbf{B\beta}.
$
For each $j$, from Theorem~\ref{thm:singlecover} we know that $\beta_j\mathbf{b}_j=\sum_{i} c_{ji} \mathbf{a}_i$ for some set of coefficients $\{c_{ji}\}_{i}$ such that $0\leq c_{ji}\leq \alpha_i$ for all $i$.
So $\mathbf{A\alpha}= \mathbf{A}(\sum_j \mathbf{c}_j)$, where $\mathbf{c}_j=(c_{j1},\cdots, c_{jk})^T$.
Since $A$ has full column rank, we have that $\alpha= \sum_j \mathbf{c}_j$. This shows that
\[\prod_{(A_i,\alpha_i)\in\calA}x^{\alpha_i}_{A_i}=\prod_{(B_j,\beta_j)\in \calB}\left(\prod_{(A_i,\alpha_i)\in \calA}x^{c_{ji}}_{A_i}\right)
\ge \prod_{(B_j,\beta_j)\in \calB}x^{\beta_j}_{B_j},\]
namely, the FNC inequality is a
nonnegative combination of at most $|\calB|$ BT inequalities.
\end{proof}

\begin{corollary}
\label{cor:BTcomb}
Suppose the FNC inequality
$\prod_{\calA}x^{\alpha_i}_{A_i}\ge \prod_{\calB}x^{\beta_j}_{B_j}$
is true for all constructible vectors,
and $m=|\calB|=1$ or $2$.
Then this inequality is
a nonnegative combination of $m$ BT inequalities.
\end{corollary}
\begin{proof}
The case $m=1$ follows from Theorem~\ref{thm:BT}.
We only need to consider the case $m=2$.
From Theorem~\ref{thm:singlecover}, we know that
$\beta_1 \mathbf{b}_1= \sum_{i} c_i \mathbf{a}_i$ for some set of coefficients $\{c_{i}\}_{i}$ such that $0\leq c_{i}\leq \alpha_i$ for all $i$.
By the definition of a FNC inequality, we know that
$\sum_{i} \alpha_i \mathbf{a}_i=\beta_1\mathbf{b}_1+\beta_2\mathbf{b}_2$.
Therefore, we have
$\beta_2 \mathbf{b}_2= \sum_{i} (\alpha_i-c_i) \mathbf{a}_i$. Consequently,
\[\prod_{(A_i,\alpha_i)\in\calA}x^{\alpha_i}_{A_i}=\left(\prod_{(A_i,\alpha_i)\in \calA}x^{c_{i}}_{A_i}\right)\left(\prod_{(A_i,\alpha_i)\in \calA}x^{\alpha_i-c_{i}}_{A_i}\right)
\ge x^{\beta_1}_{B_1}\cdot x^{\beta_2}_{B_2}.\]
This shows the FNC inequality is a
nonnegative combination of two BT inequalities.
\end{proof}

\begin{example}
Consider the NC inequality
$x_{\{1,2\}}\cdot x_{\{2,3\}}\cdot x_{\{3,4\}} \geq x_{\{1,2,3\}}\cdot x_{\{2,3,4\}}$ in Example~\ref{ex:1}.
If it is true for all constructible vectors, then from either of the above corollaries, we know that
it can be decomposed into a combination of two BT inequalities.
However, it is clear that such a decomposition does not exist.
So it is incorrect in general.
Similarly, the inequality in Example~\ref{ex:2} is incorrect.
\end{example}

\subsection{A Hybrid Approach}
In fact, none of above methods are sufficient to disprove all NC$\setminus$BT inequalities.
In this section, we demonstrate an application of the combination of these approaches.
\begin{example}
\label{ex:3}
One interesting example is the following NC inequality:
$$
x_{\{1\}}\cdot x_{\{1,2\}}\cdot x_{\{2,3\}}\cdot x_{\{3,4\}}\cdot x_{\{2,4\}} \geq x_{\{1,2,3\}}\cdot x_{\{2,3,4\}}\cdot x_{\{1,2,4\}}.
	$$
\end{example}
The example satisfies the conclusion of Theorem~\ref{thm:singlecover},
however,
we can show it is incorrect.
Our counterexample utilizes
a combination of skeleton and union-box methods.
We observe that the given inequality is a combination of
$$
x_{\{1,2\}}\cdot x_{\{2,3\}}\cdot x_{\{3,4\}} \geq x_{\{1,2,3\}}\cdot x_{\{2,3,4\}}
\quad \text{and}\quad
x_{\{1\}}\cdot x_{\{2,4\}} \geq x_{\{1,2,4\}}.
$$
We already have a skeleton counterexample for the former inequality.
The idea to is to construct an object $T$ that is the union of the skeleton and
another box $B$ such that the values of
$|T_{\{1,2\}}|,|T_{\{2,3\}}|, |T_{\{3,4\}}|, |T_{\{1,2,3\}}|, |T_{\{2,3,4\}}|$
remain (approximately) the same as that of the skeleton,
while
$|T_{\{1\}}|\cdot |T_{\{2,4\}}| \approx |T_{\{1,2,4\}}|$.
Since the skeleton allows the LHS to be arbitrarily larger than the RHS in the former inequality,
we can see that the inequality in Example~\ref{ex:3} is disproved by this object $T$.

Specifically, let $T=\SK_{M}(G)\cup B(R^3, R^{-4}, R^{-6}, R^5)$ where $G$ is the connection graph of the former inequality. As both $M$ and $R$ goes to $+\infty$ and $M=o(R)$ along the way,
we can see that $|T_{\{1\}}|\cdot |T_{\{2,4\}}| \approx |T_{\{1,2,4\}}|\approx R^4$
but
$|T_{\{1,2\}}|\approx  M+R^{-1},
|T_{\{2,3\}}|\approx M+R^{-10}, |T_{\{3,4\}}|\approx M+R^{-1},
|T_{\{1,2,3\}}|\approx M^2+R^{-7}, |T_{\{2,3,4\}}|\approx M^2+R^{-5}$.
Therefore, 
\[|T_{\{1\}}|\cdot |T_{\{2,4\}}|\cdot|T_{\{1,2\}}|\cdot|T_{\{2,3\}}|\cdot |T_{\{3,4\}}|\approx R^4\cdot M^3<R^4\cdot M^4\approx|T_{\{1,2,3\}}|\cdot|T_{\{2,3,4\}}|\cdot|T_{\{1,2,4\}}|,\]
which disproves the NC inequality in Example~\ref{ex:3}.

$\ $

In this section we have shown via different approaches that some NC$\setminus$BT inequalities are not correct.
It remains to ask whether there is a NC$\setminus$BT inequality that is correct.
We have been unable to discover one such inequality.
We have checked (in an exhaustive manner) for all inequalities in $\R^4$ by enumerating all NC inequalities and exhaustively check them on all skeletons and union of two boxes (with some pruning).
We found out that all NC$\setminus$BT inequalities in $\R^4$ are incorrect within our enumeration. We have also found that the number of NC inequalities in $\R^5$ is overwhelming for enumeration.
Hence, we propose Conjecture~\ref{conj:1} mentioned in Section~\ref{sec:intro}.

\section{Final Remarks and Acknowledgements}

All of our counterexamples in Section~\ref{sec:counterex} are essentially combinatorial,
and the constructions allow one side of the inequality to be arbitrarily larger than the other side.
We suspect that all incorrect projection inequalities can be refuted in a similar fashion.
In other words, we may not need to construct very delicate, twisted geometric objects, 
but instead just a union of a small number of boxes and skeletons to refute
any incorrect linear projection inequality.

We have developed a few other techniques to disprove some of NC inequalities.
For example, the fitting boxes model is the combination of the two models we introduced. It consists of many boxes, each constructed according to
the connection graph.
The fitting boxes model can be used to handle all $4$-dimensional inequalities.
However, it is hard to analyze and generalize to higher dimensions, and
we decided not to introduce it here.

Jian Li proposed the notion of rectangular flowers
and suspected that $\RF=\Con$, which, if true, is a natural extension of the box theorem\footnote{
Let $T$ be an object in $\mathbb{R}^m$. The box theorem states that there is a box $B$ with $|B|=|T|$
and $|B_S|\leq |T_S|$ for all $S\subseteq [m], S\ne \emptyset$.
}
 in \cite{bollobas1995projections}.
In fact, he ``verified'' the above claim empirically using hundreds of thousands datasets
(synthetically generated from different distributions with different dimensions and parameters).
Now, we know that $\RF\subsetneq\Con$.
But it is still an interesting fact that all NC inequalities are true for many ``random-like" datasets
and there may be good mathematical reasons for it.
Moreover, our counterexamples, which appear to be quite simple in retrospect,
may not be totally obvious without realizing the equivalence between rectangular flowers
and the NC inequalities.

We would like to thank
Yuval Peres for introducing BT and Shearer's inequalities to us,
Elad Verbin and Raymond Yeung for discussing non-Shannon-type inequalities, Jian Li for help formulating the problem and polishing an early version of this paper.
In particular, we would like to thank Jeff Kahn for several discussions,
and casting a doubt in the very beginning about $\RF?=\Con$, even the convexity of $\Con$, for $n\geq 4$, despite the ``empirical evidences" we showed to him.
We also thank Dan Suciu, Uri Zwick, Gil Kalai, Ely Porat, Zizhuo Wang, Chunwei Song,
Yuan Yao, Andrew Thomason and Jacob Fox  for useful discussions. Finally, we thank the four anonymous referees and the associate editor for their constructive comments and suggestions.

\bibliography{REF.bib}

\appendix
\section{Appendix 1 ($\UC=\Con$ for $n\leq 3$)}
\label{app:threed}

In this section, we prove $\mathcal{BT}_3=\Conb_3$.
This appears to be a folklore result, and we provide a proof for completeness.
Since BT inequalities are correct for all constructible vectors, it suffices to prove that, for any vector $\pi\in\R^7$, if it satisfies all BT inequalities for $3$-dimensional objects, then $\pi$ is constructible.
In fact, we show $\UC=\NC$ for $n=3$.  Since $\Con$ is sandwiched between them, all three are the same. 
To show $\UC=\NC$ for $n=3$, it suffices to show that any NC inequality in $\mathbb{R}^{3}$ is a nonnegative linear combination of BT inequalities in $\mathbb{R}^{3}$.

Let
\begin{equation}\label{eqn:NC3}
\sum_{A_i\in\calA}\pi_{A_i}\geq \sum_{B_j\in\calB} \pi_{B_j}
\end{equation}
be a NC inequality where $\calA=\{A_i\}_{i=1}^{k}$, $\calB=\{B_j\}_{j=1}^{m}$ are two families of subsets of $\{1,2,3\}$. We assume without loss of generality that $\calA\cap\calB=\emptyset$.
First of all, $\pi_{\{1\}},\pi_{\{2\}},\pi_{\{3\}}$ can only appear on the LHS of \eqref{eqn:NC3}. Assume otherwise that $\{1\}\in \calB$ then there must exist $A_i \in \calA$ such that $A_i\subseteq\{1\}$, which means $A_i=\{1\}$, a contradiction to the assumption that $\calA\cap\calB=\emptyset$. Second, if none of $\pi_{\{1,2\}},\pi_{\{1,3\}},\pi_{\{2,3\}}$ appears on the RHS of \eqref{eqn:NC3}, then the only term on the RHS of \eqref{eqn:NC3} would be $\pi_{\{1,2,3\}}$. In this case, the inequality must be a combination of BT inequalities. Third, if there exists $B_{j}\in \calB$ such that $|B_{j}|=2$, without loss of generality, we assume that $B_{1}=\{1,2\}$. By definition of NC inequalities, there must exist a $\pi_{\{1\}}$ and a $\pi_{\{2\}}$ on the LHS of \eqref{eqn:NC3}. Note that the pre-image of $(\{1,2\},1)$ in the one-to-one mapping between $\Sigma$ and $\Lambda$ must be $(\{1\},1)$ and the pre-image of $(\{1,2\},2)$ in the one-to-one mapping between $\Sigma$ and $\Lambda$ must be $(\{2\},2)$. 
Therefore, if we remove these three terms with same multiplicity from \eqref{eqn:NC3} such that after this removal there is no $\pi_{\{1,2\}}$ on the RHS, the remaining inequality must still be a NC inequality since the one-to-one mapping naturally exists. We do the same removal if $\pi_{\{1,3\}}$ or $\pi_{\{2,3\}}$ appear on the RHS of \eqref{eqn:NC3}. It can be seen that after all such removals, the RHS of the remaining inequality only contains $\pi_{\{1,2,3\}}$. Therefore, the remaining inequality is a combination of BT inequality.
Note that the three terms removed each time naturally form a BT inequality, which means the original NC inequality is a combination of several BT inequalities. Thus, we prove that $\mathcal{BT}_3=\mathcal{NC}_3=\Conb_3$.

\end{document}